\documentclass{article}
\usepackage{amsfonts, amsmath, amsthm, latexsym, amscd, graphicx, xypic, psfrag}
\usepackage{epic, eepic, epsfig, setspace, ifthen, amssymb, xy, yhmath, lscape}
\usepackage{bbold, amscd, paralist, graphics}
\usepackage{scalerel,stackengine}

\usepackage[colorlinks=true]{hyperref}

\def\br#1{\left(#1\right)}
\def\Br#1{\left[#1\right]}

\def\<#1,#2>{\langle #1,#2 \rangle}

\def\bothID{\rlap{\hbox to.97\wd0{\hss\vrule height.06\ht0 width.82\wd0}}
\copy0\rlap{\kern-.36\wd0\vrule height1.05\ht0 width.05\ht0}\kern.14\wd0}

 \DeclareMathOperator{\Var}{Var}

 \DeclareMathOperator{\Cov}{Cov}

 \DeclareMathOperator{\LGD}{LGD}

\newtheorem{theorem}{Theorem}

\newtheorem{corollary}[theorem]{Corollary}

\newtheorem{definition}[theorem]{Definition}

\newtheorem{proposition}[theorem]{Proposition}
\newtheorem{remark}[theorem]{Remark}

\stackMath
\newcommand\reallywidehat[1]{%
\savestack{\tmpbox}{\stretchto{%
  \scaleto{%
    \scalerel*[\widthof{\ensuremath{#1}}]{\kern-.6pt\bigwedge\kern-.6pt}%
    {\rule[-\textheight/2]{1ex}{\textheight}}
  }{\textheight}%
}{0.5ex}}%
\stackon[1pt]{#1}{\tmpbox}%
}
\begin{document}

\title{Credit Bubbles in Arbitrage Markets: The Geometric Arbitrage Approach to Credit Risk}

\author{Simone Farinelli\\
        Core Dynamics GmbH\\
        Scheuchzerstrasse 43\\
        CH-8006 Zurich\\
        Email: simone@coredynamics.ch\\and\\
        Hideyuki Takada\\
        Department of Information Science\\
        Narashino Campus, Toho University\\
        2-2-1-Miyama, Funabashi-Shi\\ J-274-8510 Chiba\\
        Email: hideyuki.takada@is.sci.toho-u.ac.jp
        }

\maketitle

\begin{abstract}
We apply Geometric Arbitrage Theory to obtain results in mathematical finance for credit markets, which do not need stochastic differential geometry in their formulation.
We obtain closed form equations involving default intensities and loss given defaults characterizing the no-free-lunch-with-vanishing-risk condition for corporate bonds, as well as the generic dynamics for credit market allowing for arbitrage possibilities. Moreover, arbitrage credit bubbles for both base credit assets and credit derivatives are explicitly computed for the market dynamics minimizing the arbitrage.
\end{abstract}

\tableofcontents

\section{Introduction}
This paper utilizes a conceptual structure - called in Geometric Arbitrage Theory - to
model arbitrage in credit markets. GAT embeds classical stochastic finance
into a stochastic differential geometric framework to characterize arbitrage.
 The main contribution of this approach
consists of modeling markets made of basic financial instruments
together with their term structures as principal fibre bundles.
Financial features of this market - like no arbitrage and
equilibrium - are then characterized in terms of standard
differential geometric constructions - like curvature - associated
to a natural connection in this fibre bundle.
Principal fibre bundle theory has been heavily exploited in
theoretical physics as the language in which laws of nature can be
best formulated by providing an invariant framework to describe
physical systems and their dynamics. These ideas can be carried
over to mathematical finance and economics. A market is a
financial-economic system that can be described by an appropriate
principle fibre bundle. A principle like the invariance of market
laws under change of num\'{e}raire can be seen then as gauge
invariance.\par The fact that gauge theories are the natural language
to describe economics was first proposed by Malaney and Weinstein
in the context of the economic index problem (\cite{Ma96},
\cite{We06}). Ilinski (see \cite{Il00} and \cite{Il01}) and Young
(\cite{Yo99}) proposed to view arbitrage as the curvature of a
gauge connection, in analogy to some physical theories.
Independently, Cliff and Speed (\cite{SmSp98}) further developed
Flesaker and Hughston seminal work (\cite{FlHu96}) and utilized
techniques from differential geometry (indirectly mentioned by
allusive wording) to reduce the complexity of asset models before
stochastic modeling.\par
Perhaps due to its borderline nature lying
at the intersection between stochastic finance and differential
geometry, there was almost no further mathematical research, and
the subject, unfairly considered as an exotic topic, remained
confined to econophysics, (see \cite{FeJi07}, \cite{Mo09} and
\cite{DuFiMu00}). In \cite{Fa15, Fa20} Geometric Arbitrage Theory has been
given a rigorous mathematical foundation utilizing the formal background
of stochastic differential geometry as in Schwartz (\cite{Schw80}),
Elworthy (\cite{El82}), Em\'{e}ry(\cite{Em89}), Hackenbroch and Thalmaier (\cite{HaTh94}),
Stroock (\cite{St00}) and Hsu (\cite{Hs02}). GAT can bring new insights
to mathematical finance by looking at the same concepts from a different perspective,
so that the new results can be understood without stochastic differential
geometric background. This is the case for the main contributions
of this paper, a no arbitrage characterization of credit markets.\par
More precisely,  we assume that there is a market in one currency
for both government and corporate bonds for different maturities and we
choose the government bond as num\'{e}raire.
With the formal notation introduced in Subsection \ref{CRGAT} we will prove following results.
\begin{theorem}[\textbf{No Arbitrage Credit Market}]\label{Thm1}
Let $\lambda=\lambda_t$ and $\LGD=\LGD_t$ be the default intensity and the Loss-Given-Default, respectively, of the corporate bond. Let $P^{\text{Corp, Gov}}$ and $r^{\text{Corp, Gov}}$ be the term structures and short rate for corporate and government bonds.
The following assertions are equivalent:
\begin{itemize}
\item [(i)] The credit market model satisfies the no-free-lunch-with-vanishing-risk condition.
\item[(ii)] There exists a positive local martingale $\beta=(\beta_t)_{t\ge0}$ such that deflators and short rates satisfy for  all times the condition
\begin{equation}
\boxed{
r_t^{\text{Corp}}-r_t^{\text{Gov}}=\beta_t\LGD_t\lambda_t.
}
\end{equation}
\item[(iii)] There exists a positive local martingale $\beta=(\beta_t)_{t\ge0}$ such that deflators and term structures satisfy for  all times the condition
\begin{equation}
\boxed{
\frac{P^{\text{Corp}}_{t,s}}{P^{\text{Gov}}_{t,s}}=\mathbb{E}_t\left[\exp\left(-\int_t^sdu\,\beta_u\LGD_u\lambda_u\right)\right].
}
\end{equation}
\end{itemize}
\end{theorem}
This characterization of no arbitrage credit markets has been known for a long time (see f.i. \cite{Schoe00}, page 39) and can be now easily inferred as a consequence of Geometric Arbitrage Theory.\\ Moreover,
we obtain what to our knowledge is a new result for credit markets.
\begin{theorem}[\textbf{Novikov's Condition}]\label{Thm2}
Let the credit market fullfill
\begin{equation}
\boxed{
r_t^{\text{Corp}}-r_t^{\text{Gov}}=\beta_t\LGD_t\lambda_t,
}
\end{equation}
for a positive semimartingale $(\beta_t)_{t}$ and
\begin{equation}
\boxed{
\mathbb{E}_0\left[\exp\left(\int_0^Tdt\,\frac{1}{2}\left(\frac{\lambda_t\LGD_t}{1-\LGD_tX_t}-(r_t^{\text{Corp}}-r_t^{\text{Gov}})\right)^2\frac{t}{Q_{t}(K)}\right)\right]<+\infty,
}
\end{equation}
where
\begin{equation}
\boxed{
Q_t(K):=\frac{W_t^{\dagger}W_t}{t}\sim\chi^2(K).
}
\end{equation}
Then, the credit market satisfies the no-free-lunch-with-vanishing risk.
\end{theorem}
\noindent Furthermore, applying recently discovered results about the extension of asset bubbles for markets allowing for arbitrage, in Theorem \ref{CRThm} we can compute explicitly the arbitrage bubble for credit market composed by base assets and credit derivatives. This is again a new result.\\
This paper is structured
as follows. Section $2$ reviews
classical stochastic finance and  the results of Geometric Arbitrage Theory. A guiding example is provided for a market whose asset
prices are It\^{o} processes. Proof are omitted and can be found in \cite{Fa15, Fa20}, \cite{FaTa20}, \cite{FaTa20Bis} and in \cite{FaTa20Tris}.
Section $3$ provides the mathematical background to define generalized derivatives of stochastic processes, needed in the following, since the typical processes
associated to credit risk have jumps and, in particular do not allow for Nelson's derivatives in the strong sense. Section $4$ reviews the fundamentals of credit risk and introduces the two basic model types, the structural and the reduced form (intensity based) ones. In Section $5$ the Geometric Arbitrage Theory toolbox introduced in Section $2$ is then utilize to prove results for no arbitrage credit markets. Section $6$ concludes.

\section{Geometric Arbitrage Theory Background}\label{section2}
In this section we explain the main concepts of Geometric Arbitrage Theory introduced
in \cite{Fa15, Fa20} and reviewed and extended in \cite{FaTa20}, \cite{FaTa20Bis} and \cite{FaTa20Tris}, to which we refer for proofs and examples. It can be considered as the GAT
reformulation of market risk.
\subsection{The Classical Market Model}\label{StochasticPrelude}
In this subsection we will summarize the classical set up, which
will be rephrased in section (\ref{foundations}) in differential
geometric terms. We basically follow \cite{HuKe04} and the ultimate
reference \cite{DeSc08}.\par We assume continuous time trading and
that the set of trading dates is $[0,+\infty[$. This assumption is
general enough to embed the cases of finite and infinite discrete
times as well as the one with a finite horizon in continuous time.
Note that while it is true that in the real world trading occurs at
discrete times only, these are not known a priori and can be
virtually any points in the time continuum. This motivates the
technical effort of continuous time stochastic finance.\par The
uncertainty is modelled by a filtered probability space
$(\Omega,\mathcal{A}, \mathbb{P})$, where $\mathbb{P}$ is the
statistical (physical) probability measure,
$\mathcal{A}=\{\mathcal{A}_t\}_{t\in[0,+\infty[}$ an increasing
family of sub-$\sigma$-algebras of $\mathcal{A}_{\infty}$ and
$(\Omega,\mathcal{A}_{\infty}, \mathbb{P})$ is a probability space.
The filtration $\mathcal{A}$ is assumed to satisfy the usual
conditions, that is
\begin{itemize}
\item right continuity: $\mathcal{A}_t=\bigcap_{s>t}\mathcal{A}_s$ for all $t\in[0,+\infty[$.
\item $\mathcal{A}_0$ contains all null sets of
$\mathcal{A}_{\infty}$.
\end{itemize}

The market consists of finitely many \textbf{assets} indexed by
$j=1,\dots,N$, whose \textbf{nominal prices} are given by the
vector valued semimartingale $S:[0,+\infty[\times\Omega\rightarrow\mathbf{R}^N$
denoted by $(S_t)_{t\in[0,+\infty[}$ adapted to the filtration $\mathcal{A}$.
The stochastic process $(S^ j_t)_{t\in[0,+\infty[}$ describes the
price at time $t$ of the $j$th asset in terms of  unit of cash
\textit{at time $t=0$}. More precisely, we assume the existence of a
$0$th asset, the \textbf{cash}, a strictly positive
semimartingale, which evolves according to
$S_t^0=\exp(\int_0^tdu\,r^0_u)$, where the predictable
semimartingale $(r^0_t)_{t\in[0,+\infty[}$ represents the
continuous interest rate provided by the cash account: one always
knows in advance what the interest rate on the own bank account
is, but this can change from time to time. The cash account is
therefore considered the locally risk less asset in contrast to
the other assets, the risky ones. In the following we will mainly
utilize \textbf{discounted prices}, defined as
$\hat{S}_t^j:=S_t^j/S^{0}_t$, representing the asset prices in
terms of \textit{current} unit of cash.\par
 We remark that there is no need to
assume that asset prices are positive. But, there must be at least
one strictly positive asset, in our case the cash. If we want to
renormalize the prices by choosing another asset instead of the
cash as reference, i.e. by making it to our
\textbf{num\'{e}raire}, then this asset must have a strictly
positive price process. More precisely, a generic num\'{e}raire is
an asset, whose nominal price is represented by a strictly
positive stochastic process $(B_t)_{t\in[0,+\infty[}$, and
 which is a portfolio of the original assets $j=0,1,2,\dots,N$. The discounted prices of the original
assets are  then represented in terms of the num\'{e}raire by the
semimartingales $\hat{S}_t^j:=S_t^j/B_t$.\par We assume that there
are no transaction costs and that short sales are allowed. Remark
that the absence of transaction costs can be a serious limitation
for a realistic model. The filtration $\mathcal{A}$ is not
necessarily generated by the price process
$(S_t)_{t\in[0,+\infty[}$: other sources of information than
prices are allowed. All agents have access to the same information
structure, that is to the filtration $\mathcal{A}$.\par A
 \textbf{strategy} is a predictable stochastic
process $x:[0,+\infty[\times\Omega\rightarrow\mathbf{R}^N$
describing the portfolio holdings. The stochastic process $(x^
j_t)_{t\in[0,+\infty[}$ represents the number of pieces of $j$th
asset portfolio held by the portfolio as time goes by. Remark that
the It\^{o} stochastic integral
\begin{equation}
\int_0^tx\cdot dS=\int_0^tx_u\cdot dS_u,
\end{equation}
\noindent and the Stratonovich stochastic integral
\begin{equation}\label{strat}
\int_0^tx\circ dS:=\int_0^tx\cdot d
S+\frac{1}{2}\int_0^td\left<x,S\right>=\int_0^tx_u\cdot d
S_u+\frac{1}{2}\int_0^td \left<x,S\right>_u
\end{equation}
 are well defined
for this choice of integrator ($S$) and integrand ($x$), as long as
the strategy is \textbf{admissible}, meaning that it is $v$-admissible for some $v\ge0$, that is,  $x=(x_t)_{t\in[0,+\infty[}$ is a predictable semimartingale for which the It\^{o} integral satisfies $\int_0^tx\cdot dS\ge-v$ for all $t\ge0$.
 The bracket $\left<\cdot,
\cdot\right>$ in (\ref{strat}) denotes the quadratic covariation of two processes. In a general context
strategies do not need to be semimartingales, but if we want the quadratic covariation and hence
 Stratonovich's integral to be well defined, we must require this additional assumption.
For details about stochastic integration we refer to Appendix A in
\cite{Em89}, which summarizes Chapter VII of the authoritative
\cite{DeMe80}. The portfolio value is the process
$\{V_t\}_{t\in[0,+\infty[}$ defined by
\begin{equation}
V_t:=V_t^x:=x_t\cdot S_t.
\end{equation}
An admissible strategy $x$ is said to be \textbf{self-financing}
if and only if the portfolio value at time $t$ is given by
\begin{equation}
V_t=V_0+\int_0^tx_u\cdot dS_u.
\end{equation}
\noindent This means that the portfolio gain is  the It\^{o}
integral of the strategy with the price process as integrator: the
change of portfolio value is purely due to changes of the assets'
values. The self-financing condition can be rewritten in
differential form as
\begin{equation}
dV_t=x_t\cdot dS_t.
\end{equation}
As pointed out in \cite{BjHu05}, if we want to utilize
Stratonovich's integral to rephrase the self-financing condition,
while maintaining its economical interpretation (which is necessary
for the subsequent constructions of mathematical finance), we write
\begin{equation}
V_t=V_0+\int_0^tx_u\circ dS_u-\frac{1}{2}\int_0
^td\left<x,S\right>_u
\end{equation}
or, equivalently
\begin{equation}
dV_t=x_t\circ dS_t-\frac{1}{2}\,d\left<x,S\right>_t.
\end{equation}
\par An \textbf{arbitrage strategy} (or arbitrage for short) for the market model is an admissible self-financing
strategy $x$, for which one of the following condition holds for
some horizon $T>0$:
\begin{itemize}
\item $P[V_0^{x}<0]=1$ and $P[V_T^{x}\ge0]=1$,
\item $P[V_0^{x}\le0]=1$ and $P[V_T^{x}\ge0]=1$ with $P[V_T^{x}>0]>0$.
\end{itemize}
In Chapter 9 of \cite{DeSc08} the no arbitrage condition is given
a topological characterization. In view of the fundamental Theorem
of asset pricing, the no-arbitrage condition is substituted by a
stronger condition, the so called
no-free-lunch-with-vanishing-risk.

\begin{definition}[\textbf{Arbitrage}] Let the process $(S_t)_{[0,+\infty[}$ be a semimartingale and $(x_t)_{t\in[0,+\infty[}$ be admissible self-financing strategy. Let us consider trading up to time $T\le\infty$. The portfolio wealth at time $t$ is given by by $V_{t}(x):=V_0+\int_0^tx_u\cdot dS_u$, and we denote by $K_0$ the subset of $L^0(\Omega, \mathcal{A}_{T},P)$ containing all such $V_T(x)$, where $x$ is any admissible self-financing strategy.
We define
\begin{itemize}
\item $C_0:=K_0-L_+^0(\Omega, \mathcal{A}_{T},P)$.
\item $C:=C_0\cap L_+^{\infty}(\Omega, \mathcal{A}_{T},P)$.
\item $\bar{C}$: the closure of $C$ in $L^{\infty}$ with respect to the norm topology.
\item $\mathcal{V}^{V_0}:=\left\{(V_{t})_{t\in[0,+\infty[}\,\big{|}\, V_t=V_t(x), \,\text{where } x \text{ is } V_0\text{-admissible} \right\}$.
\item $\mathcal{V}_T^{V_0}:=\left\{V_T\,\big{|}\,(V_{t})_{t\in[0,+\infty[}\in\mathcal{V}^{V_0}\right\}$: terminal wealth for $V_0$-admissible self-financing strategies.
\end{itemize}
We say that $S$ satisfies
\begin{itemize}
\item \textbf{(NA), no arbitrage}, if and only if $C \cap L^{\infty}(\Omega, \mathcal{A}_{T},P)=\{0\}$.
\item \textbf{(NFLVR), no-free-lunch-with-vanishing-risk},  if and only if $\bar{C} \cap L^{\infty}(\Omega, \mathcal{A}_{T},P)=\{0\}$.
\item \textbf{(NUPBR), no-unbounded-profit-with-bounded-risk}, if and only if $\mathcal{V}_T^{V_0}$ is bounded in $L^0$ for some $V_0>0$.
\end{itemize}
\end{definition}
\noindent The relationship between these three different types of arbitrage has been elucidated in \cite{DeSc94} and in \cite{Ka97} with the proof of the following result.
\begin{theorem}
\begin{equation}
\boxed{
\text{(NFLVR)}\Leftrightarrow \text{(NA)}+\text{(NUPBR)}.
}
\end{equation}
\end{theorem}

\noindent Delbaen and Schachermayer proved in 1994 (see
\cite{DeSc08} Chapter 9.4, in particular the main Theorem 9.1.1)
\begin{theorem}[\textbf{First Fundamental Theorem of Asset Pricing in Continuous
Time}]\label{ThmDeSch} Let $(S_t)_{t\in[0,+\infty[}$ and
$(\hat{S}_t)_{t\in[0,+\infty[}$ be  bounded semimartingales. There
is an equivalent martingale measure $\mathbb{P}^*$ for the
discounted prices $\hat{S}$ if and only if the market model
satisfies the (NFLVR).
\end{theorem}
This is a generalization for continuous time of the
Dalang-Morton-Willinger Theorem proved in 1990 (see \cite{DeSc08},
Chapter 6) for the discrete time case, where the (NFLVR) is relaxed
to the (NA) condition. The Dalang-Morton-Willinger Theorem
generalizes to arbitrary probability spaces the Harrison and Pliska
Theorem (see \cite{DeSc08}, Chapter 2) which holds true in discrete
time for finite probability spaces.\par An equivalent alternative to
the martingale measure approach for asset pricing purposes is given
by the pricing kernel (state price deflator) method.
\begin{definition}
Let $(S_t)_{t\in[0,+\infty[}$  be a semimartingale describing the
price process for the assets of our market model. The positive
semimartingale $(\beta_t)_{t\in[0,+\infty[}$ is called
\textbf{pricing kernel (or state price deflator)} for $S$ if and only if
$(\beta_tS_t)_{t\in[0,+\infty[}$ is a $\mathbb{P}$-martingale.
\end{definition}
As shown in \cite{HuKe04} (Chapter 7, definitions 7.18, 7.47 and
Theorem 7.48), the existence of a pricing kernel is equivalent to
the existence of an equivalent martingale measure for a specific choice of num\'{e}raire.
If we want the num\'{e}raire to be arbitrary, like the one we originally choose for the model, then we have to additionally assume that the pricing kernel $\beta$ is a local $\mathbb{P}$-martingale.
\begin{theorem}\label{ThmZ}
Let $(S_t)_{t\in[0,+\infty[}$ and $(\hat{S}_t)_{t\in[0,+\infty[}$
be  bounded semimartingales. The process $\hat{S}$ admits an
equivalent martingale measure $\mathbb{P}^*$ if and only if there
is a pricing kernel $\beta$ for $S$ (or for $\hat{S}$), which is a local martingale.
\end{theorem}

\subsection{Geometric Reformulation of the Market Model: Primitives}
We are going to introduce a more general representation of the
market model introduced in section \ref{StochasticPrelude}, which
better suits to the arbitrage modeling task.
\begin{definition}\label{defi1}
A \textbf{gauge} is an ordered pair of two $\mathcal{A}$-adapted
real valued semimartingales $(D, P)$, where
$D=(D_t)_{t\ge0}:[0,+\infty[\times\Omega\rightarrow\mathbf{R}$ is
called \textbf{deflator} and
$P=(P_{t,s})_{t,s}:\mathcal{T}\times\Omega\rightarrow\mathbf{R}$,
which is called \textbf{term structure}, is considered as a
stochastic process with respect to the time $t$, termed
\textbf{valuation date} and
$\mathcal{T}:=\{(t,s)\in[0,+\infty[^2\,|\,s\ge t\}$. The parameter
$s\ge t$ is referred as \textbf{maturity date}. The following
properties must be satisfied a.s. for all $t, s$ such that $s\ge
t\ge 0$:
 \begin{itemize}
  \item [(i)] $P_{t,s}>0$,
  \item [(ii)] $P_{t,t}=1$.
 \end{itemize}
\end{definition}

\begin{remark}
Deflators and term structures can be considered
\textit{outside the context of fixed income.} An arbitrary
financial instrument is mapped to a gauge $(D, P)$ with the
following economic interpretation:
\begin{itemize}
\item Deflator: $D_t$ is the value of the financial instrument at time $t$ expressed in terms of some num\'{e}raire. If we
choose the cash account, the $0$-th asset as num\'{e}raire, then
we can set $D_t^j:=\hat{S}_t^j=\frac{S_t^j}{S_t^0}\quad(j=1,\dots
N)$.
\item Term structure: $P_{t,s}$ is the value at time $t$ (expressed in units of
deflator at time $t$) of a synthetic zero coupon bond with
maturity $s$ delivering one unit of financial instrument at time
$s$. It represents a term structure of forward prices with respect
to the chosen num\'{e}raire.
\end{itemize}
\noindent We point out that there is no unique choice for
deflators and term structures describing an asset model. For
example, if a set of deflators qualifies, then we can multiply
every deflator  by the same positive semimartingale to obtain
another suitable set of deflators. Of course term structures have
to be modified accordingly. The term "deflator" is clearly
inspired by actuarial mathematics. In the present context it
refers to a nominal asset value up division by a strictly positive
semimartingale (which can be the state price deflator if this
exists and it is made to the num\'{e}raire). There is no need to
assume that a deflator is a positive process. However, if we want
to make an asset to our num\'{e}raire, then we have to make sure
that the corresponding deflator is a strictly positive stochastic
process.
\end{remark}

\begin{definition}\label{int}
The term structure can be written as a functional of the
\textbf{instantaneous forward rate } f defined as
\begin{equation}
\boxed{
  f_{t,s}:=-\frac{\partial}{\partial s}\log P_{t,s},\quad
  P_{t,s}=\exp\left(-\int_t^sdhf_{t,h}\right).
  }
\end{equation}
\noindent and
\begin{equation}
\boxed{
 r_t:=\lim_{s\rightarrow t^+}f_{t,s}
 }
\end{equation}
\noindent is termed \textbf{short rate}.
\end{definition}
\begin{remark}

Since $(P_{t,s})_{t,s}$ is a $t$-stochastic process (semimartingale) depending on a parameter $s\ge t$, the $s$-derivative can be defined deterministically, and the expressions above make sense pathwise in a both classical and generalized sense. In a generalized sense we will always have a $\mathcal{D}^{\prime}$ derivative for any $\omega\in \Omega$; this corresponds to a classic $s$-continuous  derivative if $P_{t,s}(\omega)$ is a $C^1$-function of $s$ for any fixed $t\ge0$ and $\omega\in\Omega$.
\end{remark}

\begin{remark} The special choice of vanishing interest rate $r\equiv0$ or flat term structure
$P\equiv1$ for all assets corresponds to the classical model,
where only asset prices and their dynamics are relevant.
\end{remark}

\noindent Let us consider -in continuous time- a market with $N$ assets and a
num\'{e}raire. A general portfolio at time $t$ is described by the
vector of nominals $x\in \mathfrak{X}$, for an open set
$\mathfrak{X}\subset\mathbb{R}^N$. Following Definition \ref{defi1}, the asset model induces for $j=1,\dots,N$ the gauge
\begin{equation}(D^j,P^j)=((D_t^j)_{t\in[0, +\infty[},(P_{t,s}^j)_{s\ge t}),\end{equation}
\noindent where $D^j$ denotes the deflator and $P^j$ the term
structure. This can be written as
\begin{equation}P_{t,s}^j=\exp\left(-\int_t^sf^j_{t,u}du\right),\end{equation}
where $f^j$ is the instantaneous forward rate process for the $j$-th asset and the corresponding short rate is given by $r_t^j:=\lim_{u\rightarrow 0^+}f^j_{t,u}$. For a
portfolio with nominals $x\in \mathfrak{X}\subset\mathbb{R}^N$ we define
\begin{equation}
\boxed{
D_t^x:=\sum_{j=1}^Nx_jD_t^j\quad
f_{t,u}^x:=\sum_{j=1}^N\frac{x_jD_t^j}{\sum_{j=1}^Nx_jD_t^j}f_{t,u}^j\quad
P_{t,s}^x:=\exp\left(-\int_t^sf^x_{t,u}du\right).
}
\end{equation}
The short rate writes
\begin{equation}
\boxed{
r_t^x:=\lim_{u\rightarrow 0^+}f^x_{t,u}=\sum_{j=1}^N\frac{x_jD_t^j}{\sum_{j=1}^Nx_jD_t^j}r_t^j.
}
\end{equation}
The image space of all possible strategies reads
\begin{equation}M:=\{(t,x)\in [0,+\infty[\times\mathfrak{X}\}.\end{equation}

\noindent We can prove following results which characterizes arbitrage as
curvature.
\begin{theorem}[\textbf{No Arbitrage}]\label{holonomy}
The following assertions are equivalent:
\begin{itemize}
\item [(i)] The market model satisfies the no-free-lunch-with-vanishing-risk condition.
\item[(ii)] There exists a positive local martingale $\beta=(\beta_t)_{t\ge0}$ such that deflators and short rates satisfy for all times and all portfolio nominals $(t,x)\in M$ the condition
\begin{equation}
\boxed{
r_t^x=-\mathcal{D}\log(\beta_tD_t^x).
}
\end{equation}
\item[(iii)] There exists a positive local martingale $\beta=(\beta_t)_{t\ge0}$ such that deflators and term structures satisfy for all times and all portfolio nominals $(t,x)\in M$ the condition
\begin{equation}
\boxed{
P^x_{t,s}=\frac{\mathbb{E}_t[\beta_sD^x_s]}{\beta_tD^x_t}.
}
\end{equation}
\end{itemize}
\end{theorem}

\noindent This motivates the following definition.
\begin{definition}
The market model satisfies the  \textbf{zero curvature condition (ZC)} if and only if
the curvature vanishes a.s.
\end{definition}

\noindent As proved in \cite{FaTa20}, the two weaker notions of arbitrage, the zero curvature and the no-unbounded-profit-with-bounded-risk are equivalent. Together with the well-knowns results in \cite{DeSc94} and \cite{Ka97} this leads to
\begin{theorem}\label{thm_ZC_equiv}
\begin{equation}
\boxed{
\text{(NFLVR)}\Leftrightarrow \left\{
                                \begin{array}{ll}
                                  \text{(NUPBR)}\Rightarrow \text{(ZC)}\\
                                  \text{(NA)}
                                \end{array}
                              \right.
}
\end{equation}
\end{theorem}

\noindent As an example to demonstrate how the most
important geometric concepts of Section \ref{section2} can be
applied we consider an asset model whose dynamics is given by a multidimensional
multidimensional It\^{o}-process. Let us consider a market consisting of  $N+1$
assets labeled by $j=0,1,\dots,N$, where the $0$-th asset is the
cash account utilized as a num\'{e}raire. Therefore, as explained
in the introductory Subsection \ref{StochasticPrelude}, it
suffices to model the price dynamics of the other assets
$j=1,\dots,N$ expressed in terms of the $0$-th asset. As vector
valued semimartingales for the discounted price process
$\hat{S}:[0,+\infty[\times\Omega\rightarrow\mathbf{R}^N$ and the short rate $r:[0,+\infty[\times\Omega\rightarrow\mathbf{R}^N$, we chose
the
 multidimensional It\^{o}-processes given by
\begin{equation}\label{Dyn}
\begin{split}
d\hat{S}_t&=\hat{S}_t(\alpha_tdt+\sigma_tdW_t)\\
dr_t&=a_tdt+b_tdW_t,
\end{split}
\end{equation}
where
\begin{itemize}
\item $(W_t)_{t\in[0,+\infty[}$ is a standard $P$-Brownian motion in
$\mathbf{R}^K$, for some $K\in\mathbf{N}$, and,
\item $(\sigma_t)_{t\in[0,+\infty[}$,  $(\alpha_t)_{t\in[0,+\infty[}$ are  $\mathbf{R}^{N\times
K}$-, and respectively,  $\mathbf{R}^{N}$- valued stochastic processes,
\item $(b_t)_{t\in[0,+\infty[}$,  $(a_t)_{t\in[0,+\infty[}$ are  $\mathbf{R}^{N\times
L}$-, and respectively,  $\mathbf{R}^{N}$- valued  stochastic processes.
\end{itemize}

\begin{proposition}\label{PropIto}
Let the dynamics of a market model be specified by following It\^o's processes as in (\ref{Dyn}), where we additionally assume that  the coefficients
\begin{itemize}
\item $(\alpha_t)_t,(\sigma_t)_t$, and  $(r_t)_t$ satisfy
\begin{equation}
\lim_{s\rightarrow t^+}\mathbb{E}_s[\alpha_t]=\alpha_t,\quad\lim_{s\rightarrow t^+}\mathbb{E}_s[r_t]=r_t,\quad\lim_{s\rightarrow t^+}\mathbb{E}_s[\sigma_t]=\sigma_t,
\end{equation}
\item $(\sigma_t)_t$ is an It\^o's process,
\item $(\sigma_t)_t$ and $(W_t)_t$ are independent processes.
\end{itemize}
Then, the market model satisfies the (ZC) condition if and only if
\begin{equation}\label{ZCCond}
\alpha_t+r_t \in {\rm Range} (\sigma_t).
\end{equation}
\end{proposition}
\begin{remark}In the case of the classical model, where there are no term structures (i.e. $r\equiv0$), the condition (\ref{ZCCond}) reads as $\alpha_t\in{\rm Range} (\sigma_t)$.
\end{remark}
\noindent The stronger (NFLVR) condition in this guiding example appears in the following result.
\begin{proposition}\label{NovikovThm}
For the market model whose dynamics is specified by the SDEs
\begin{equation}
\begin{split}
d\hat{S}_t&=\hat{S}_t(\alpha_tdt+\sigma_tdW_t)\\
dr_t&=a_tdt+b_tdW_t,
\end{split}
\end{equation}
the no-free-lunch-with-vanishing risk condition (NFLVR) is satisfied if \textbf{Novikov's condition}
\begin{equation}\label{Novikov}
\mathbb{E}_0\left[\exp\left(\int_0^T\frac{1}{2}\left|{\sigma_t}^{\dagger}(\sigma_t{\sigma_t}^{\dagger})^{-1}(\alpha_t+r_t)\right|^2du\right)\right]<+\infty,
\end{equation}
is fulfilled.
\end{proposition}

\subsection{Bubbles in Arbitrage Markets}\label{bubbles}
Asset bubbles were first introduced in \cite{JPS10} for complete markets and have been recently extended to and computed for arbitrage markets in \cite{FaTa20Bis} and in \cite{FaTa20Tris}, whose main findings we briefly summarize here below.

\begin{definition} The \textbf{cash flow bundle} is defined as
\begin{equation}
\boxed{
\mathcal{V}:= \underbrace{[0,T]\times \mathfrak{X}}_{=M} \times \mathbb{R}^{[0,+\infty[}.
}
\end{equation}
\noindent The space of the sections of the cashflow bundle can be made into a scalar product space by introducing,
for stochastic sections $f=f(t, x, \omega)=(f_s(t, x, \omega))_{s\in[0,+\infty[}$ and $g=g(t, x, \omega)=(g_s(t, x, \omega))_{s\in[0,+\infty[}$
\begin{equation}
\begin{split}
&(f,g):=\int_{\Omega}dP\int_Xd^Nx\int_0^Tdt\left<f,g\right>(t, x, \omega)=\mathbb{E}_0\left[(f,g)_{L^2(M,\mathbb{R}^{[0,+\infty[})}\right]=(f,g)_{L^2(\Omega,\mathcal{V},\mathcal{A}_0,dP)},\\
&\text{where}\\
&\left<f,g\right>(t, x, \omega):=\int_0^{+\infty}dsf_s(t, x, \omega)g_s(t,x, \omega).
\end{split}
\end{equation}
The Hilbert space of integrable sections reads
\begin{equation}
\mathcal{H}:=L^2(\Omega,\mathcal{V},\mathcal{A}_0,dP)=\left\{\left.f=f(t,x,\omega)=(f_s(t,x,\omega))_{s\in[0,+\infty[}\right|\,(f,f)_{L^2(\Omega,\mathcal{V},\mathcal{A}_0,dP)}<+\infty\right\}.
\end{equation}
\noindent Let us extend the coordinate vector $x\in\mathbb{R}^N$ with a $0$th component given by the time $t$.
Let $X=\sum_{j=0}^NX_j\frac{\partial}{\partial x_j}$ be a vector field over $M$ and $f=(f_s)_s$ a section of the cashflow bundle $\mathcal{V}$. Then,
\begin{equation}
\boxed{
\nabla^{\mathcal{V}}_Xf_t:=\sum_{j=0}^N\left(\frac{\partial f_t}{\partial x_j}+K_jf_t\right)X_j,
}
\end{equation}
where
\begin{equation}
\boxed{
K_0(x)=-r_t^x\qquad
K_j(x)=\frac{D^j_t}{D_t^x}\quad(1\le j\le N).
}
\end{equation}
\noindent defines a covariant derivative on the cash flow bundle $\mathcal{V}$.
\end{definition}

\begin{definition}[\textbf{Spectral Lower Bound}]
The highest spectral lower bound of the connection Laplacian on the cash flow bundle $\mathcal{V}$ is given by
\begin{equation}
\lambda_0:=\inf_{\substack{\varphi\in C^{\infty}(M,\mathcal{V})\\\varphi\neq 0\\B_N(\varphi)=0}}\frac{(\nabla^{\mathcal{V}}\varphi, \nabla^{\mathcal{V}}\varphi)_{\mathcal{H}}}{(\varphi,\varphi)_{\mathcal{H}}}
\end{equation}
and it is assumed on the subspace
\begin{equation}
E_{\lambda_0}:=\left\{\varphi\left|\,\varphi\in C^{\infty}(M,\mathcal{V})\cap\mathcal{H}, \, B_N(\varphi)=0,\, (\nabla^{\mathcal{V}}\varphi, \, \nabla^{\mathcal{V}}\varphi)_{\mathcal{H}}\ge \lambda_0 (\varphi,\varphi)_{\mathcal{H}}\right.\right\}.
\end{equation}
\noindent The space
\begin{equation}
\mathcal{K}_{\lambda_0}:=\{\varphi\in E_{\lambda_0}\left|\,\varphi\ge0,\,\mathbb{E}[\varphi]=1\right.\}
\end{equation}
contains all candidates for the Radon-Nikodyn derivative
\begin{equation}\label{risk_neutral_measure}
\frac{d\mathbb{P}^*}{d\mathbb{P}}=\varphi,
\end{equation}
for a probability measure $P^*$ absolutely continuous with respect to the statistical measure $P$.
\end{definition}
\noindent Theorem \ref{ThmDeSch}, that is the first fundamental theorem of asset pricing can be reformulated as
\begin{proposition}
The market model satisfies the (NFLVR) condition if and only if $\lambda_0=0$. Any probability measure defined by (\ref{risk_neutral_measure}) with $\varphi\in\mathcal{K}_{0}$ is a risk neutral measure, that is  $(D_t)_{t\in[0,T]}$ is a vector valued martingale with respect to $P^*$, i.e.
\begin{equation}
\mathbb{E}^*_t[D_s]
=D_t\qquad\text{ for all }s\ge t \text{ in }[0,T].
\end{equation}
\noindent The market is complete if and only if $\lambda_0=0$ and $\dim E_{0}=1$.
\end{proposition}
\noindent For arbitrage markets we have that $\lambda_0>0$ and there exists no risk neutral probability measures. Nevertheless it is possible to define a fundamental value, however not in a unique way.

\begin{definition}[\textbf{Basic Assets' Arbitrage Fundamental Prices and Bubbles}]\label{arbitrage_bubble}
Let $(C_t)_{t\in[0,T]}$ the $\mathbb{R}^N$ cash flow stream stochastic process associated to the $N$ assets of the market model with given spectral lower bound $\lambda_0$ and Radon-Nikodym candidates' subspace $\mathcal{K}_{\lambda_0}$.
For a given choice of $\varphi\in\mathcal{K}_{\lambda_0}$ the approximated fundamental value of the assets with stochastic $\mathbb{R}^N$-valued price process $(S_t)_{t\in[0,T]}$ is defined as
\begin{equation}
\boxed{
S_t^{*,\varphi}:=\mathbb{E}_t\left[\varphi\left(\int_t^\tau dC_u\,\exp\left(-\int_t^ur_s^0ds\right)+S_{\tau}\exp\left(-\int_t^\tau r_s^0ds\right)\,1_{\{\tau<+\infty\}}\right)\right]\,1_{\{t<\tau \}},
}
\end{equation}
\noindent where $\tau$ denotes the maturity time of all risky assets in the market model, and, the approximated bubble is defined as
\begin{equation}
\boxed{
B_t^{\varphi}:=S_t-S_t^{*,\varphi}.
}
\end{equation}
The fundamental price vector for the assets and their asset bubble prices are defined as
\begin{equation}\label{phi_0}
\boxed{
\begin{split}
S^*_t&:=S_t^{*,\varphi_0}\\
B_t&:=B_t^{\varphi_0}\\
\varphi_0&:=\arg\min_{\varphi\in \mathcal{K}_{\lambda_0}}\mathbb{E}_0\left[\int_0^Tds\,|B_s^{\varphi}|^2\right].
\end{split}
}
\end{equation}
\noindent The probability measure $P^*$ with Radon-Nikodym derivative
\begin{equation}
\frac{dP^*}{dP}=\varphi_0
\end{equation}
is termed \textbf{minimal arbitrage measure}.
\end{definition}

\begin{proposition}
The assets' fundamental values can be expressed as conditional expectation with respect to the minimal arbitrage measure as
\begin{equation}\label{fund}
\boxed{
S^*_t:=\mathbb{E}_t^*\left[\int_t^\tau dC_u\,\exp\left(-\int_t^ur_s^0ds\right)+S_\tau\exp\left(-\int_t^\tau r_s^0ds\right)\,1_{\{\tau<+\infty\}}\right]\,1_{\{t<\tau \}}.
}
\end{equation}
\end{proposition}
\noindent Formula (\ref{fund}) can be reformulated in terms of the curvature, by means of which we can extend
 Jarrow-Protter-Shimbo's results in \cite{JPS10} to the following bubble decomposition and classifications theorems proved in \cite{FaTa20Bis}.
\begin{theorem}[\textbf{Bubble decomposition and types}]\label{arbitrage_bubble_contingent}
Let $T=+\infty$ and $\tau$ denote the maturity time of all risky assets in the market model. $S_t$ admits a unique (up to $P$-evanescent set) decomposition
\begin{equation}
S_t=\tilde{S}_t+B_t,
\end{equation}
\noindent where $B=(B_t)_{t\in[0,T]}$ is a c\`{a}dl\`{a}g process satisfying
for all $j=1,\dots,N$
\begin{equation}\label{bubble_value}
\boxed{
B_t^j=S_t^j-\mathbb{E}_t^*\left[\int_t^\tau dC_u^j\,\exp\left(-\int_t^uds\, r^0_s\right)+\exp\left(\int_t^\tau ds\,r_s^0\right)S_{\tau}^j1_{\{\tau<+\infty\}}\right]\,1_{\{t<\tau\}},
}
\end{equation}
\noindent into a sum of fundamental and bubble values.\\
If there exists a non-trivial bubble $B_t^j$ in an asset's price for $j=1,\dots,N$, then, there exists a probability measure $P^*$ equivalent to $P$, for which we have three and only three possibilities:
\begin{itemize}
\item[\textbf{Type 1:}] $B_t^j$ is local super- or submartingale with respect to both $P$ and $P^*$, if $P[\tau=+\infty]>0$.
\item[\textbf{Type 2:}] $B_t^j$ is local super- or submartingale  with respect to both $P$ and $P^*$, but not uniformly integrable super- or submartingale, if $B_t^j$ is unbounded but with $P[\tau<+\infty]=1$.
\item[\textbf{Type 3:}] $B_t^j$ is a strict local super- or sub- $P$- and $P^*$-martingale, if $\tau$ is a bounded stopping time.
\end{itemize}
\end{theorem}
\noindent Next we analyze the situation for derivatives.
\begin{definition}[\textbf{Contingent Claim's Arbitrage Fundamental Price and Bubble}]
Let us consider in the context of Definition (\ref{arbitrage_bubble}) a European option given by the contingent claim  with a unique pay off $G(S_T)$ at time $T$ for an appropriate real valued function $G$ of $N$ real variables. The contingent claim fundamental price and its corresponding arbitrage bubble is defined in the case of base assets paying no dividends as
\begin{equation}\label{option_price}
\boxed{
\begin{split}
V^*_t(G)&:=\mathbb{E}_t\left[\varphi_0\exp\left(-\int_t^Tr_s^0ds\right)G(S_T)\,1_{\{T<+\infty\}}\right]1_{\{t<T\}}=\\
&=\mathbb{E}^*\left[\exp\left(-\int_t^Tr_s^0ds\right)G(S_T)\,1_{\{T<+\infty\}}\right]1_{\{t<T\}}\\
B_t(G)&:=V_t(G)-V^*_t(G),
\end{split}
}
\end{equation}
\noindent where $\varphi_0$ is the minimizer for the basic assets bubbled defined in (\ref{phi_0}), $P^*$ the minimal arbitrage measure and $(V_t(G))_{t\in[0,T]}$ is the price process of the European option.\\\
In the case of base assets paying dividends the definition becomes
\begin{equation}\label{option_price_div}
\boxed{
\begin{split}
V^*_t(G)&:=\mathbb{E}_t\left[\varphi_0\exp\left(-\int_t^Tr_s^0ds\right)G\left(S_T\exp\left(\frac{C_T}{S_T}(T-t)\right)\right)\,1_{\{T<+\infty\}}\right]1_{\{t<T\}}=\\
&=\mathbb{E}^*\left[\exp\left(-\int_t^Tr_s^0ds\right)G\left(S_T\exp\left(\frac{C_T}{S_T}(T-t)\right)\right)\,1_{\{T<+\infty\}}\right]1_{\{t<T\}}\\
B_t(H)&:=V_t(G)-V^*_t(G),
\end{split}
}
\end{equation}
\noindent where $\frac{C_t^j}{S_t^j}$ is the instantaneous dividend rate for the $j$-th asset.
\end{definition}

\begin{remark}
If the market is complete, then $\lambda_0=0$ and $\mathcal{K}_{\lambda_0}=\{\varphi_0\}$, where $\varphi_0$ is the Radon-Nykodim derivative of the unique risk neutral probability measure with respect to the statistical probability measure. The definitions in (\ref{arbitrage_bubble}) and in (\ref{arbitrage_bubble_contingent}) coincide for the complete market with the definitions of fundamental value and asset bubble price for both base asset and contingent claim introduced by Jarrow, Protter and Shimbo in \cite{JPS10}, proving that they are a natural extension to markets allowing for arbitrage opportunities.
\end{remark}

\begin{corollary}\label{bubble_valuation}
The bubble discounted values for the base assets in Definition \ref{arbitrage_bubble}) and for the contingent claim on the base assets paying dividends in Definition \ref{arbitrage_bubble_contingent}
\begin{equation}
\widehat{B}_t:=\exp\left(-\int_0^tr_s^0ds\right)B_t\qquad \widehat{B}(G)_t:=\exp\left(-\int_0^tr_s^0ds\right)B(G)_t
\end{equation}
\noindent satisfy the equalities
\begin{equation}\label{cont_value_bubble}
\boxed{
\begin{split}
&\widehat{B}_t^j=D_t^j-\left(\mathbb{E}_t^*\left[D_{\tau}^j 1_{\{\tau<+\infty\}}\right]+\mathbb{E}_t^*\left[\widehat{C}_\tau^j 1_{\{\tau<+\infty\}}\right]-\widehat{C}_t^j\right) 1_{\{t<\tau\}}\\
&\widehat{B}_t(G)=\widehat{V}_t(G)-\mathbb{E}_t^*\left[\widehat{G}\left(S_T\exp\left(\frac{C_T}{S_T}(T-t)\right)\right)1_{\{T<+\infty\}}\right]1_{\{t<T\}}.
\end{split}
}
\end{equation}
\noindent where
\begin{equation}
\widehat{C}_t^j:=\exp\left(-\int_0^tr_s^0ds\right)C_t^j\qquad\widehat{G}:=\exp\left(-\int_0^Tr_s^0ds\right)G\qquad\widehat{V}_t(G):=\exp\left(-\int_0^tr_s^0ds\right)V_t(G)
\end{equation}
\noindent are the discounted cashflow for the $j$-th asset, the discounted contingent claim payoff, and the discounted value of the derivative.

\end{corollary}
\begin{theorem} The following statements hold true for any market model with $T\le+\infty$ allowing for arbitrage:
\begin{itemize}
 \item[$(a)$] Market portfolio, asset values and term structures solving the minimal arbitrage problem are identically distributed along time, their returns are centered and serially uncorrelated:
 \begin{equation}\label{id}
 \boxed{
 \begin{split}
 &([x_t;D_t;r_t]))_{t\in[0,T]}\text{ is an i.d. process with respect to the statistical probability measure $P$,}\\
 &([\mathcal{D}x_t;\mathcal{D}D_t;\mathcal{D}r_t])_{t\in[0,T]}\text{ is centered and has vanishing autocovariance function,}
 \end{split}
 }
 \end{equation}
 \noindent In particular, conditional and total expectations of asset values, nominals and term structures are constant over time:
\begin{equation}
\boxed{
 \begin{array}{lll}
 \mathbb{E}_0[x_t]\equiv\text{const}\qquad &\mathbb{E}_0[D_t]\equiv\text{const}\qquad &\mathbb{E}_0[r_t]\equiv\text{const}\qquad\\
 \mathbb{E}_0[\mathcal{D}x_t]\equiv0\qquad &\mathbb{E}_0[\mathcal{D}D_t]\equiv0\qquad &\mathbb{E}_0[\mathcal{D}r_t]\equiv0.\qquad
 \end{array}
 }
 \end{equation}

 \noindent The variances of portfolio nominals are concurrent with those of the deflators:
\begin{equation}\label{vols}
\boxed{
\Var_0\left(D_t^j\right)\Var_0\left(\frac{x_t^j}{x_t\cdot D_t}\right)\ge\frac{1}{4},
}
\end{equation}
\noindent for all indices $j=1,\dots,N$.

 \item[$(b)$] Expectation and variance of the bubble discounted value for the $j$-th asset read
 \begin{equation}\label{exp_var_base_bubble}
 \boxed{
  \begin{split}
  \mathbb{E}_0[\widehat{B}_t^j]&=\mathbb{E}_0\left[D_t^j-\widehat{C}_t^j\right]-\mathbb{E}_0^*\left[D_T^j - \widehat{C}_T^j\right]\\
  \Var_0(\widehat{B}_t^j)&=\Var_0\left(D_t^j-\widehat{C}_t^j\right)+\Var_0^*\left(D_T^j-\widehat{C}_T^j\right)-2\Cov_0^*\left(D_t^j-\widehat{C}_t^j,D_T^j-\widehat{C}_T^j\right).
  \end{split}
  }
 \end{equation}

 \item[$(c)$] Expectation and variance of the bubble discounted value for the contingent claim $G(S_T)$ on the base assets read
 \begin{equation}\label{exp_var_contingent_bubble}
 \boxed{
  \begin{split}
  \mathbb{E}_0[\widehat{B}_t(G)]&=\mathbb{E}_0[\widehat{V}_t(G)]-\mathbb{E}_0^*\left[\widehat{G}\left(S_T\exp\left(\frac{C_T}{S_T}(T-t)\right)\right)\right]\\
  \Var_0(\widehat{B}_t(G))&=\Var_0(\widehat{V}_t(G))+\Var_0^*\left(\widehat{G}\left(S_T\exp\left(\frac{C_T}{S_T}(T-t)\right)\right)\right).
  \end{split}
  }
 \end{equation}
\end{itemize}
\end{theorem}

\section{Generalized Derivatives of Stochastic
Processes}\label{Derivatives} In stochastic differential geometry one would like to lift
the constructions of stochastic analysis from open subsets of
$\mathbf{R}^N$ to  $N$ dimensional differentiable manifolds. To that
aim, chart invariant definitions are needed and hence a stochastic
calculus satisfying the usual chain rule and not It\^{o}'s Lemma is required,
(cf. \cite{HaTh94}, Chapter 7, and the remark in Chapter 4 at the
beginning of page 200). That is why the papers about geometric arbitrage theory are mainly concerned in
 by stochastic integrals and derivatives meant in \textit{Stratonovich}'s
sense and not in \textit{It\^{o}}'s. Of course, at the end of the computation, Stratonovich integrals can be transformed into It\^{o}'s.
Note that a fundamental portfolio equation, the self-financing condition cannot be directly formally expressed with Stratonovich integrals, but first with It\^{o}'s and then transformed into Stratonovich's, because it is a non-anticipative condition.
\begin{definition}\label{Nelson}
Let $I$ be a real interval and $Q=(Q_t)_{t\in I}$ be a  $\mathbb{R}^N$-valued stochastic process on the probability space
$(\Omega, \mathcal{A}, P)$. The process $Q$ determines three families of $\sigma$-subalgebras of the $\sigma$-algebra $\mathcal{A}$:
\begin{itemize}
\item[(i)] ''Past'' $\mathcal{P}_t$, generated by the preimages of Borel sets in $\mathbf{R}^N$  by all mappings $Q_s:\Omega\rightarrow\mathbf{R}^N$ for $0<s<t$.
\item[(ii)] ''Future'' $\mathcal{F}_t$, generated by the preimages of Borel sets in $\mathbf{R}^N$  by all mappings $Q_s:\Omega\rightarrow\mathbf{R}^N$ for $0<t<s$.
\item[(iii)] ''Present'' $\mathcal{N}_t$, generated by the preimages of Borel sets in $\mathbf{R}^N$  by the mapping $Q_s:\Omega\rightarrow\mathbf{R}^N$.
\end{itemize}
Let $Q=(Q_t)_{t\in I}$ be continuous.
 Assuming that the following limits exist,
\textbf{Nelson's stochastic derivatives} are defined as
\begin{equation}
\boxed{
\begin{split}
&DQ_t:=\lim_{h\rightarrow
0^+}\mathbb{E}\Br{\left.\frac{Q_{t+h}-Q_t}{h}\right|
\mathcal{P}_t}\text{: forward derivative,}\\
& D_*Q_t:=\lim_{h\rightarrow
0^+}\mathbb{E}\Br{\left.\frac{Q_{t}-Q_{t-h}}{h}\right|
\mathcal{F}_{t}}\text{: backward derivative,}\\
&\mathcal{D}Q_t:=\frac{DQ_t+D_*Q_t}{2}\text{: mean derivative}.
\end{split}
}
\end{equation}
Let $\mathcal{S}^1(I)$ the set of all processes $Q$ such that
$t\mapsto Q_t$, $t\mapsto DQ_t$ and $t\mapsto D_*Q_t$ are continuous
mappings from $I$ to $L^2(\Omega, \mathcal{A})$. Let
$\mathcal{C}^1(I)$ the completion of $\mathcal{S}^1(I)$ with respect
to the norm
\begin{equation}
\boxed{
\|Q\|:=\sup_{t\in I}\br{\|Q_t\|_{L^2(\Omega, \mathcal{A})}+\|DQ_t\|_{L^2(\Omega, \mathcal{A})}+\|D_*Q_t\|_{L^2(\Omega, \mathcal{A})}}.
}
\end{equation}
\end{definition}

\begin{remark}
The stochastic derivatives $D$, $D_*$ and  $\mathcal{D}$
correspond to It\^{o}'s, to the anticipative and, respectively,  to Stratonovich's integral (cf. \cite{Gl11}). The process space $\mathcal{C}^1(I)$ contains all It\^{o} processes. If $Q$ is a Markov process, then the sigma algebras $\mathcal{P}_t$ (''past'') and $\mathcal{F}_t$ (''future'') in the definitions of forward and backward derivatives can be substituted by the sigma algebra $\mathcal{N}_t$ (''present''), see Chapter 6.1 and 8.1 in (\cite{Gl11}).
\end{remark}

\noindent Stochastic derivatives can be defined pointwise in $\omega\in\Omega$ outside the class $\mathcal{C}^1$ in terms of generalized functions.
\begin{definition}
Let $Q:I\times\Omega\rightarrow\mathbb{R}^N$ be a continuous linear functional in the test processes $\varphi:I\times\Omega\rightarrow\mathbb{R}^N$ for $\varphi(\cdot,\omega)\in C^{\infty}_c(I,\mathbb{R}^N)$. We mean by this that for a fixed $\omega\in\Omega$ the functional $Q(\cdot,\omega)\in\mathcal{D}(I,\mathbb{R}^N)$, the topological vector space of continuous distributions. We can then define
\textbf{Nelson's generalized stochastic derivatives:}
\begin{equation}
\boxed{
\begin{split}
&DQ(\varphi_t):=-Q(D\varphi_t)\text{: forward generalized derivative,}\\
& D_*Q(\varphi_t):=-Q(D_*\varphi_t)\text{: backward generalized derivative,}\\
&\mathcal{D}(\varphi_t):=-Q(\mathcal{D}\varphi_t)\text{: mean generalized derivative}.
\end{split}
}
\end{equation}
\end{definition}
\noindent If the generalized derivative is regular, then the process has a derivative in the classic sense. This construction is nothing else than a straightforward pathwise lift of the theory of generalized functions to a wider class stochastic processes which do not a priori allow for Nelson's derivatives in the strong sense. We will utilize this feature in the treatment of credit risk, where many processes with jumps occur.

\section{Credit Risk}
After having introduced the machinery of Geometric Arbitrage Theory we can tackle the modeling of assets' defaults and their recoveries.
\subsection{Classical Credit Risk Models}\label{classicCredit}
Here we summarize the standard ways to model credit risk. We follow \cite{JaPr04} and \cite{FrSc11}. There are basically two possibilities for modeling defaults: structural model types on one hand and reduced form (intensity based) model types on the other. The difference between them can be characterized in terms of the information assumed known by the observer. Structural models assume that the observer has the same information set as the firm's manager, i.e. the complete knowledge of all firm's assets and liabilities. In most situations, this knowledge leads to a predictable default time. In contrast, reduced form models assume that the observer has the same information set as the market, i.e. an incomplete knowledge of the firm's condition. In most cases, this imperfect knowledge leads to an inaccessible default time.\par As highlighted in \cite{JaPr04} these models are not disconnected and disjoint model types as it was commonly supposed, but rather they are really the same model containing different informational assumptions. The key distinction between structural and reduced form models is not in the characteristic of the default time (predictable vs. inaccessible), but in the information set available to the observer. Indeed, structural models can be transformed into reduced form models as the information set changes and becomes less refined, from that observable by the firm's manager to that which is observed by the market.\par Rather than comparing model types on the basis of their forecasting performance, the model type choice should be based on the information set available by the observer. For general risk management purposes the relevant set is the information available in the market, hence a structural model is to be preferred. By contrast, if one is interested in a firm's risky debt or related credit derivatives, then reduced form models are the better approach.\par
Let us introduce the standard setup by utilizing the market model introduced in Subsection \ref{section2} to account for defaults and different information sets. Credit risk management investigates an entity (corporation, bank, individual) that borrows funds, promises to return these funds under a prespecified contractual agreement, and which may default before the funds (in their entirety) are repaid. That for, we introduce a market allowing for two kind of assets (beside the cash account), non-defaultable (e.g. government bonds) and defaultable ones (e.g. corporate bonds).
\begin{definition}[\textbf{Information Structures}]
To model uncertainty, there are two filtrations for $(\Omega,\mathcal{A}, \mathbb{P})$:
\begin{itemize}
\item \textbf{Market Filtration: } This is the $\mathcal{A}=\{\mathcal{A}_t\}_{t\in[0,+\infty[}$ used so far for market risk, representing the information available by all market participants.
\item \textbf{Global Filtration: } This is the $\mathcal{G}=\{\mathcal{G}_t\}_{t\in[0,+\infty[}$ representing the information available by the management of the bond issuer company.
\end{itemize}
The global filtration is postulated to contain the market filtration. i.e. $\mathcal{A}_t\subset\mathcal{G}_t$ for all $t\ge0$. Unless otherwise specified conditional probabilities and expectations refer to the market filtration, i.e. $\mathbb{P}_t[\cdot]=\mathbb{P}[\cdot|\mathcal{A}_t]$ and $\mathbb{E}_t[\cdot]:=\mathbb{E}[\cdot|\mathcal{A}_t]$.
\end{definition}

\begin{definition}[\textbf{Default and Recovery Models}] Let $D_t^{\text{Corp}}$ be the market value of a defaultable asset.
\begin{itemize}
\item \textbf{Default Indicator:}
 \begin{equation}
 \boxed{
   X_t:=\left\{
      \begin{array}{ll}
        1, & \hbox{corporate bond in default state at time $t$} \\
        0, & \hbox{corporate bond in non-default state at time $t$.}
      \end{array}
    \right.
    }
 \end{equation}
\item \textbf{Time-To-Default:}
 \begin{equation}
 \boxed{
   \tau:=\inf\{t\ge0\,|\,X_t=1\}.
   }
 \end{equation}
\item \textbf{Conditional Default Probability:}
 \begin{equation}
 \boxed{
  p_{t,s}^{\mathcal{A}}:=\mathbb{P}_t\left[\tau\le s\,|\,\tau>t\right].
  }
 \end{equation}
\item \textbf{Structural Model:} Let $(E_t)_{t\ge0}$ be the corporate equity process with default threshold $E_{\min}$. The structural model for default is the following specification for the default indicator:
    \begin{equation}
    \boxed{
    X_t:=1_{\{E_t\le E_{\min}\}}.
    }
    \end{equation}
   The corporate equity dynamics is observable in the market, i.e. $\mathcal{A}_t\supset\sigma\left(\left\{E_s\,|\,s\le t\right\}\right)$, and it is typically given by an It\^{o}'s diffusion with respect to the market filtration
    \begin{equation}
    \boxed{
     dE_t=E_t(\alpha_t^E(E_t)+\sigma_t^E(E_t))dW_t.
     }
    \end{equation}
\item \textbf{Intensity Model:} The global filtration $\mathcal{G}$ contains the filtration $\sigma\left(\left\{\tau,Y_s\,|\,s\le t\right\}\right)$ generated by the Time-To-Default and by a vector of state variables $Y_t$, which follows an It\^{o}'s diffusion. The default indicator is a Cox process induced by $\tau$ with a positive intensity process $(\lambda_t)_{t\ge0}$, which corresponds to the following specification:
    \begin{equation}\label{inte}
    \boxed{
    X_t:=1_{\{\Lambda^{-1}(E)\le t\}},
    }
    \end{equation}
  where $\Lambda_t:=\int_0^tdh\lambda_h$ and $E\sim\text{Exp}(1)$ is an exponentially distributed random variable.
\item \textbf{Loss-Given-Default:} If there is default at time $t$, then the recovered value at time $t^+$ is given by $(1-\LGD_t)D_{t^-}^{\text{Corp}}$. The stochastic process $(\LGD_t)_{t\ge0}$ is observable in the market filtration.
\end{itemize}
\end{definition}

\begin{proposition}\label{IntStruct}
The default probabilities in the two models read:
\begin{itemize}
\item \text{Structural Model:}
    \begin{equation}
    \boxed{
     p_{t,s}^{\mathcal{A}}=\mathbb{P}_t\left[E_s\le E_{\min}\,|\,E_t\ge E_{\min}\right].
     }
    \end{equation}

\item \text{Intensity Model:}
    \begin{equation}
    \boxed{
    p_{t,s}^{\mathcal{G}}=1-\mathbb{E}\left[\left.\exp\left(-\int_t^sdh\lambda_h\right)\right|\mathcal{G}_t\right].
    }
    \end{equation}
\end{itemize}
\end{proposition}
A known fact about structural credit risk models is summarized by the following proposition.
\begin{proposition}\label{pred}
In the structural models Time-To-Default is a predictable stopping time and corresponds to the first hitting time of the barrier
\begin{equation}
\boxed{
\tau=\inf\{t\ge0\,|\,E_t\le E_{\min}\}.
}
\end{equation}
\end{proposition}

\begin{remark}
A stopping time $\tau$ is a non-negative random variable such that the event $\{\tau\le t\}\in\mathcal{A}_t$ for every $t\ge0$. A stopping time is predictable if there exists a sequence of stopping times $(\tau_n)_{n\ge0}$ such that $\tau_n$ is increasing with $n$, $\tau_n<\tau$ for all $n\ge0$ and $\lim_{n\rightarrow+\infty}\tau_n=\tau$ almost surely. Intuitively, an event described by a predictable stopping time is ''known'' to occur ''just before'' it happens, since it is announced by an increasing sequence of stopping times. This is certainly the situation for  structural models with respect to the market filtration. In essence, altough default is an uncertain event and thus technically a surprise, it is not a ''true surprise'' to the global observer, because it can be anticipated with almost certainty by watching the path of company equity value. The key characteristic of a structural model is the observability of the market information set $\mathcal{A}_t\supset\sigma\left(\left\{E_s\,|\,s\le t\right\}\right)$ and not the fact that default is predictable.
\end{remark}
\noindent Another known fact about reduced form credit risk models (cf. \cite{JaPr04}) is
\begin{proposition}\label{unpred}
In the reduced form models Time-To-Default is a totally inaccessible stopping time, i.e. for every predictable stopping time $S$ the event $\{\omega\in\Omega\,|\,\tau(\omega)=S(\omega)<+\infty\}$ vanishes almost surely.
\end{proposition}
Now, what are the relationships between structural and reduced form models?
The reason for the transformation of the default time $\tau$ from a predictable stopping time in Proposition \ref{pred} to an inaccessible stopping time in Proposition \ref{unpred} is that between the time observations of the company equity value, we do not know how the equity value has evolved. Consequently, prior to our next observation, default could occur unexpectedly (as a complete surprise). If one changes the information set held by the observer from more to less information from $\mathcal{G}$ to $\mathcal{A}$, then a structural model with default being a predictable stopping time can be transformed into a hazard rate model with default being an inaccessible stopping time:
\begin{equation}
\boxed{
    \mathbb{E}\left[\left.p_{t,s}^{\mathcal{G}}\right|\mathcal{A}_t\right]=1-\mathbb{E}_t\left[\exp\left(-\int_t^sdu\,\lambda_u\right)\right]=1-\exp\left(-\int_t^sdu\,h_u\right),
    }
\end{equation}
where $h$ denotes the deterministic \textbf{hazard function}.
Thus, the overall relevant structure is that of the two filtrations and how stopping times behave in them. The structural models play a role in the determination of the structure generating the default time. But as soon as the information available to the observer is reduced or obscured, one needs to project onto a smaller filtration, and then the default time becomes totally inaccessible, and the compensator $\Lambda$ of the one jump point process $1-X_t$ becomes the object of interest. If the compensator can be written in the form $\Lambda_t=\int_0^tdh\lambda_h$, then the process $(\lambda_t)_{t\ge0}$ can be interpreted as the instantaneous rate of default, given the observer's information set. In that case, from Proposition \ref{IntStruct}, we derive
\begin{proposition}\label{prop38}
Structural and intensity models are related by the following relationship
\begin{equation}
\boxed{
\lambda_t=\lim_{s\rightarrow t^+}\frac{\partial}{\partial s}\mathbb{P}_t\left[E_s\le E_{\min}\,|\,E_t> E_{\min}\right].
}
\end{equation}
\end{proposition}

\begin{proposition}
For both structural and reduced for credit model, if the market model satisfies the no-arbitrage-with-vanishing-risk condition,
the risk free discounted value of the corporate bond reads for any $s\ge t$
\begin{equation}
\boxed{
D^{\text{Corp}}_t=\mathbb{E}_t^*\left[\left((1-\LGD_\tau)1_{\{\tau\le s\}}+1_{\{\tau>s\}}\right)D_s^{\text{Corp}}\right].
}
\end{equation}
\end{proposition}
\begin{proof}
Let $S_t$ denote the value of the corporate bond and $c_t$ its cash flow intensity. Then
\begin{equation}
S_t=\mathbb{E}^*_t\left[\int_t^{+\infty}dh\,c_h
\exp\left(-\int_t^hdu\,r^0_u\right)\right].
\end{equation}
Therefore,
\begin{equation}
\begin{split}
D_t^{\text{Corp}}&=\mathbb{E}^*_t\left[\int_t^{+\infty}dh\,c_h
\exp\left(-\int_t^hdu\,r^0_u\right)\right]=\\
&=\mathbb{E}^*_t\left[\int_t^{+\infty}dh\,\delta(h-\tau)
\left((1-\LGD_h)1_{\{h\le s\}}+1_{\{h>s\}}\right)D_s^{\text{Corp}}\right]=\\
&=\mathbb{E}_t^*\left[\left((1-\LGD_\tau)1_{\{\tau\le s\}}+1_{\{\tau>s\}}\right)D_s^{\text{Corp}}\right]].
\end{split}
\end{equation}
\end{proof}
Is it possible to characterize the model type on the basis of Nelson's differentiation property of the default indicator?
\begin{proposition}\label{propX}
In the structural model the generalized Nelson forward derivative of the default indicator reads
\begin{equation}
\boxed{
D^{\mathcal{A}}X_t=\lim_{s\rightarrow t^+}\frac{\partial}{\partial s}\mathbb{P}_t\left[E_s\le E_{\min}\,|\,E_t> E_{\min}\right].
}
\end{equation}
\end{proposition}
\begin{proof}
The default probability can be developed as
\begin{equation}
\begin{split}
\mathbb{P}_t\left[E_s\le E_{\min}\,|\,E_t> E_{\min}\right]&=\frac{\mathbb{E}_t\left[1_{\{E_s\le E_{\min}\}}1_{\{E_t> E_{\min}\}}\right]}{\mathbb{E}_t\left[1_{\{E_t> E_{\min}\}}\right]}=\\
&=\mathbb{E}_t\left[1_{\{E_s\le E_{\min}\}}\right].
\end{split}
\end{equation}
Therefore, we obtain
\begin{equation}
\begin{split}
&\lim_{s\rightarrow t^+}\frac{\partial}{\partial s}\mathbb{P}_t\left[E_s\le E_{\min}\,|\,E_t> E_{\min}\right]=\\
&=\lim_{s\rightarrow t^+}\lim_{h\rightarrow0^+}\frac{\mathbb{P}_t\left[E_{s+h}\le E_{\min}\,|\,E_t> E_{\min}\right]-\mathbb{P}_t\left[E_s\le E_{\min}\,|\,E_t> E_{\min}\right]}{h}=\\
&=\lim_{s\rightarrow t^+}\lim_{h\rightarrow 0^+}\mathbb{E}_t\left[\frac{1_{\{E_{s+h}\le E_{\min}\,|\,E_t> E_{\min}\}}-1_{\{E_s\le E_{\min}\,|\,E_t> E_{\min}\}}}{h}\right]=\\
&=\lim_{s\rightarrow t^+}D^{\mathcal{A}}1_{\{E_s\le E_{\min}\,|\,E_t> E_{\min}\}}=D^{\mathcal{A}}X_t,
\end{split}
\end{equation}
where Nelson's derivative $D$ has to be understood in the generalized sense.\\
\end{proof}
\begin{proposition}\label{prop41}
For the intensity model, with respect to the global filtration, the Nelson forward generalized derivative of the default indicator reads
\begin{equation}
\boxed{
D^{\mathcal{G}}X_t=\lambda_t.
}
\end{equation}
\end{proposition}

\begin{proof}
We can reformulate (\ref{inte})  as a structural model as
\begin{equation}
X_t=1_{\{\Lambda^{-1}(E)\le t\}}=1_{\{-\Lambda_t+(E)\le 0\}},
\end{equation}
and therefore
\begin{equation}
p_{t,s}^{\mathcal{G}}=\mathbb{P}\left[-\Lambda_s+E\le 0\,|\,\{-\Lambda_t+E> 0\}\cap\mathcal{G}_t\right]
\end{equation}
By mimicking the proof of Proposition \ref{propX} with these adaptations we have
\begin{equation}
\begin{split}
D^{\mathcal{G}}X_t&=\lim_{s\rightarrow t^+}\frac{\partial}{\partial s}\mathbb{P}\left[-\Lambda_s+E\le 0\,|\,\{-\Lambda_t+E> 0\}\cap\mathcal{G}_t\right]=\\
&=\lim_{s\rightarrow t^+}\frac{\partial}{\partial s}\left(1-\mathbb{E}\left[\left.\exp\left(-\int_t^sdh\lambda_h\right)\right|\mathcal{G}_t\right]\right)\\
&=\lim_{s\rightarrow t^+}\mathbb{E}\left[\left.\exp\left(-\int_t^sdh\lambda_h\right)\lambda_s\right|\mathcal{G}_t\right]=\mathbb{E}[\lambda_t|\mathcal{G}_t]=\lambda_t.
\end{split}
\end{equation}
\end{proof}
\noindent Therefore, by comparing Propositions \ref{prop38}, \ref{propX} and \ref{prop41}, we can conclude that
\begin{theorem}
Structural models admit an intensity formulation where the intensity is given by the Nelson forward derivative with respect to the market filtration.
\end{theorem}

\subsection{Geometric Arbitrage Theory Credit Risk Model}\label{CRGAT}
Now can carry out the analysis of credit markets described in Subsection \ref{classicCredit} by utilizing the tools of Geometric Arbitrage Theory introduced in Section \ref{section2} and, in particular, Proposition \ref{NovikovThm}.
\begin{definition}[\textbf{Credit Market}]\label{CM} A (simple) credit market consists in  a \textbf{government asset} $(S_t^{\text{Gov}})_{t\in[0,T]}$ with cash flow $(C_t^{\text{Gov}})_{t\in[0,T]}$  and a \textbf{corporate asset} $(S_t^{\text{Corp}})_{t\in[0,T]}$ with cash flow $(C_t^{\text{Corp}})_{t\in[0,T]}$ . The \textbf{credit asset} is defined as a portfolio consisting in a long position in the corporate asset and in a short position in the government asset: $S_t^{\text{Cred}}:=S_t^{\text{Corp}}-S_t^{\text{Cred}}$ and $C_t^{\text{Cred}}:=C_t^{\text{Corp}}-C_t^{\text{Cred}}$. Following Definition \ref{defi1}, let $(D^{\text{Gov}}, P^{\text{Gov}})$ and $(D^{\text{Corp}}, P^{\text{Corp}})$ be the gauge corresponding to the government and, respectively, the corporate asset, with their corresponding term structures. The credit gauge $(D^{\text{Cred}}, P^{\text{Cred}})$ is defined as
\begin{itemize}
\item \textbf{Deflator:} $D^{\text{Cred}}_t:=D^{\text{Corp}}_t-D^{\text{Gov}}_t:=\exp\left(\int_0^tds\,r_s^0\right)(S^{\text{Corp}}_t-S^{\text{Gov}}_t)$,
\item \textbf{Discounted Cash Flow:} $\widehat{C}^{\text{Cred}}_t:=\exp\left(\int_0^tds\,r_s^0\right)(\widehat{C}^{\text{Corp}}_t-\widehat{C}^{\text{Gov}}_t)$,
\item \textbf{Instantaneous Forward Rate:} $f^{\text{Cred}}_{t,s}:=f^{\text{Corp}}_{t,s}-f^{\text{Gov}}_{t,s}$,
\item \textbf{Short Rate:} $r^{\text{Cred}}_t:=\lim_{s\rightarrow  t^+}f^{\text{Cred}}_{t,s}$,
\item \textbf{Term Structure:} $P^{\text{Cred}}_{t,s}:=\exp\left(-\int_t^sdhf^{\text{Cred}}_{t,h}\right)$.

\end{itemize}
The credit gauge represents all relevant information necessary to model a credit market for bonds with arbitrary maturities and of a given rating in one currency. Different ratings correspond to different credit gauges. In the vector notation of Definitions \ref{defi1} and \ref{int} we have, with the choice $x^{\text{Cred}}:=[-1,+1]^{\dagger}$,
\begin{equation}
\begin{aligned}
&D_t:=[D_t^{\text{Gov}}, D_t^{\text{Corp}}]^{\dagger} &r_t:=[r_t^{\text{Gov}}, r_t^{\text{Corp}}]^{\dagger}\\
&D_t^{\text{Cred}}=D_t^{x^{\text{Cred}}}& r_t^{\text{Cred}}=r_t^{x^\text{Cred}}.
\end{aligned}
\end{equation}
\end{definition}
\begin{proposition}
The credit asset gauge satisfies following properties:
\begin{itemize}
\item Deflator:
 \begin{equation}
  D^{\text{Cred}}_t=(1-\LGD_tX_t)D^{\text{Corp}}_0-D^{\text{Gov}}_{t}.
\end{equation}
\item Term structure:
 \begin{equation}
  P^{\text{Cred}}_{t,s}=\frac{P^{\text{Corp}}_{t,s}}{P^{\text{Gov}}_{t,s}}.
 \end{equation}
 \item Short rate:
 \begin{equation}
  r^{\text{Cred}}_t=r^{\text{Corp}}_t-r^{\text{Gov}}_t.
 \end{equation}
\end{itemize}
\end{proposition}

\noindent We can apply Theorem \ref{holonomy} to the credit market to characterize no arbitrage.
\begin{theorem}[\textbf{No Arbitrage Credit Market}]\label{CreditHolonomy}
Let $\lambda=\lambda_t$ and $\LGD=\LGD_t$ be the default intensity and the Loss-Given-Default, respectively, of the corporate bond. The following assertions are equivalent:
\begin{itemize}
\item [(i)] The credit market model satisfies the no-free-lunch-with-vanishing-risk condition.
\item[(ii)] There exists a positive local martingale $\beta=(\beta_t)_{t\ge0}$ such that deflators and short rates satisfy for  all times the condition
\begin{equation}r_t^{\text{Cred}}=\beta_t\LGD_t\lambda_t.\end{equation}
\item[(iii)] There exists a positive local martingale $\beta=(\beta_t)_{t\ge0}$ such that deflators and term structures satisfy for  all times the condition
\begin{equation}P^{\text{Cred}}_{t,s}=\mathbb{E}_t\left[\exp\left(-\int_t^sdu\,\beta_u\LGD_u\lambda_u\right)\right].\end{equation}
\end{itemize}
\end{theorem}
\begin{proof}
By Theorem \ref{holonomy} (iii) for $N=2$, government ($x=[1,0]^\dagger$) and corporate ($x=[0,-1]^\dagger$) the  (NFLVR) condition reads
\begin{equation}
\left\{
\begin{split}
&P_{t,s}^{\text{Gov}}=\frac{\mathbb{E}_t[\beta_sD_s^{\text{Gov}}]}{\beta_tD_t^{\text{Gov}}}\\
&P_{t,s}^{\text{Corp}}=\frac{\mathbb{E}_t[\beta_sD_s^{\text{Corp}}]}{\beta_tD_t^{\text{Corp}}}\\
\end{split}
\right.
\end{equation}
By making the Government asset to the num\'{e}raire (i.e. $D_t^{\text{Gov}}\equiv1$), we obtain from the first equation $P_{t,s}^{\text{Gov}}\equiv1$ and hence $r^{\text{Gov}}_t\equiv0$. We define an equivalent martingale measure for $(D_t)_t$ by setting the Radon-Nykodym derivative as
\begin{equation}
\frac{dP^*}{dP}:=\beta_T,
\end{equation}
and rewrite the second equation utilizing $\beta_t=\mathbb{E}_t\left[\frac{dP^*}{dP}\right]$ as
\begin{equation}
P_{t,s}^{\text{Corp}}=\frac{\mathbb{E}_t[\beta_sD_s^{\text{Corp}}]}{\beta_tD_t^{\text{Corp}}}=\mathbb{E}_t^*[1-\LGD_sX_s].
\end{equation}
By taking on both sides of the equation $-\left.\frac{\partial}{\partial s}\right|_{s:=t}$, we obtain
\begin{equation}
r_t^{\text{Corp}}=\LGD_t\lambda^*_t,
\end{equation}
because, by rewriting the default indicator by means of the Heaviside function and the time-to-default,
\begin{equation}
\begin{split}
&\LGD_t=\LGD_\tau\Theta(t-\tau)\\
&X_t=\Theta(t-\tau)\\
&\mathbb{E}_t^*[D^{\mathcal{A}}(\LGD_tX_t)]=\mathbb{E}_t^*[\underbrace{D^{\mathcal{A}}(\LGD_t)X_t)}_{=\delta(t-\tau)\Theta(t-\tau)=0}+\LGD_tD^{\mathcal{A}}(X_t)]=\LGD_t\lambda_t^*,
\end{split}
\end{equation}
\noindent where $\lambda_t^*$ is the default intensity with respect to $P^*$.
Therefore, with the government asset as a num\'{e}raire
\begin{equation}
r_t^{\text{Cred}}=r_t^{\text{Corp}}=\LGD_t\lambda^*_t=\beta_t\LGD_t\lambda_t,
\end{equation}
which is (ii), to which (iii) is equivalent. The proof is completed.\\
\end{proof}
\noindent Theorem \ref{Thm1} follows directly from Theorem \ref{CreditHolonomy}. We can now apply Proposition \ref{PropIto} and Proposition \ref{NovikovThm}  to the credit market to find the dynamics satisfying the no-free-lunch-with-vanishing-risk condition.

We have a version where the equivalence of (NFLVR) with (ZC) holds true if Novikov's growth condition for the instantaneous Sharpe ratio is satified.
\begin{corollary}\label{Novikov2}
For the market with the government bond chosen as num\'{e}raire and a corporate bond dynamics $(D_t^{\text{Corp}})_t$ specified by the weak limits $D_t^{\text{Corp}}=\mathcal{D}^{\prime}-\lim_{\epsilon\rightarrow 0}D_t^{\text{Corp}, \epsilon}$ and $r_t^{\text{Corp}}=\mathcal{D}^{\prime}-\lim_{\epsilon\rightarrow 0}r_t^{\text{Corp}, \epsilon}$ of  It\^{o} processes $(D_t^{\text{Corp},\epsilon})_t$  $(r_t^{\text{Corp},\epsilon})_t$ and satisfying SDE
\begin{equation}\label{credAssDyn}
\begin{split}
&dD_t^{\text{Corp},\epsilon}=D_t^{\text{Corp},\epsilon}(\alpha_t^{\text{Corp},\epsilon}dt+\sigma_t^{\text{Corp},\epsilon}dW_t)\\
&dr_t^{\text{Corp},\epsilon}=a_t^{\text{Corp},\epsilon}dt+b_t^{\text{Corp},\epsilon}dW_t
\end{split}
\end{equation}
where
\begin{itemize}
\item $(W_t)_{t\in[0,+\infty[}$ is a standard $P$-Brownian motion in
$\mathbf{R}^K$, for some $K\in\mathbf{N}$,
\item ${(\alpha_t^{\text{Corp},\epsilon})}_{t\in[0,+\infty[}$, ${(\sigma_t^{\text{Corp},\epsilon})}_{t\in[0,+\infty[}$ and $(r_t^{\text{Corp},\epsilon})_{t\in[0,+\infty[}$   are  $\mathbf{R}$-, $\mathbf{R}^K$- and, respectively, $\mathbf{R}^K$-valued predictable
 stochastic processes,
\item $(\alpha_t)_t,(\sigma_t)_t$, and  $(r_t)_t$ satisfy
\begin{equation}
\lim_{s\rightarrow t^+}\mathbb{E}_s[\alpha_t^{\text{Corp},\epsilon}]=\alpha_t^{\text{Corp},\epsilon},\quad\lim_{s\rightarrow t^+}\mathbb{E}_s[r_t^{\text{Corp},\epsilon}]=r_t^{\text{Corp},\epsilon},\quad\lim_{s\rightarrow t^+}\mathbb{E}_s[\sigma_t^{\text{Corp},\epsilon}]=\sigma_t^{\text{Corp},\epsilon},
\end{equation}
\item $(\sigma_t^{\text{Corp},\epsilon})_t$ is an It\^o's process,
\item $(\sigma_t^{\text{Corp},\epsilon})_t$ and $(W_t)_t$ are independent processes.
\end{itemize}
\noindent the no-free-lunch-with-vanishing-risk condition is satisfied if Novikov's condition is satisfied,
which is the case if and only if the zero-curvature condition is satisfied and
\begin{equation}\label{NovikovLGD}
\mathbb{E}_0\left[\exp\left(\int_0^Tdt\,\frac{1}{2}\left(\frac{\lambda_t\LGD_t}{1-\LGD_tX_t}-r_t^{\text{Corp}}\right)^2\frac{t}{Q_{t}(K)}\right)\right]<+\infty,
\end{equation}
where
\begin{equation}
Q_t(K):=\frac{W_t^{\dagger}W_t}{t}\sim\chi^2(K),
\end{equation}
is a chi-squared distributed real random variable.
\end{corollary}
\noindent Theorem \ref{Thm2} follows from Corollary \ref{Novikov2}, because any $\mathcal{D}^{\prime}$ process can be regularized by a sequence of It\^{o}'s processes satisfying the assumptions of the corollary as
\begin{equation}
\lim_{\epsilon\rightarrow0}D^{\text{Corp},\epsilon}_t(\varphi)= D^{\text{Corp}}_t(\varphi)\quad\text{ and }\quad\lim_{\epsilon\rightarrow0}r^{\text{Corp},\epsilon}_t(\varphi)= r^{\text{Corp}}_t(\varphi
\end{equation}
for all $\varphi(\cdot,\omega)\in C^{\infty}_c([0,T],\mathbb{R})$
\begin{proof}
The only thing to prove is inequality (\ref{NovikovLGD}). On one hand
\begin{equation}\label{eq1}
\left\{
  \begin{array}{ll}
    D_t^{\text{Gov}}\equiv1\\
    D_t^{\text{Corp}}=(1-\LGD_tX_t)D_0^{\text{Corp}}.
  \end{array}
\right.
\end{equation}
On the other, the solution of (\ref{credAssDyn}) reads
\begin{equation}
D_t^{\text{Corp},\epsilon}=D_0^{\text{Corp},\epsilon}\exp\left(\int_0^t\left(\alpha_u^{\text{Corp},\epsilon}+\frac{1}{2}\sigma_u^{\text{Corp},\epsilon}{\sigma_u^{\text{Corp},\epsilon}}^{\dagger}\right)du+\int_0^t\sigma_u^{\text{Corp},\epsilon}dW_u\right),
\end{equation}
and, hence,
\begin{equation}\label{dlog}
\mathcal{D}\log(D_t^{\text{Corp},\epsilon})=\alpha_t^{\text{Corp},\epsilon}+\frac{1}{2}\sigma_t^{\text{Corp},\epsilon}{\sigma_t^{\text{Corp},\epsilon}}^{\dagger}+\sigma_t^{\text{Corp},\epsilon}\frac{W_t}{2t},
\end{equation}
because
\begin{equation}\label{eqDD}
\mathcal{D}W_t=\frac{W_t}{2t}\quad\text{ and }\quad \left<\sigma^{\text{Corp},\epsilon},W\right>_t\equiv0.
\end{equation}
Since the zero-curvature condition is satisfied  we infer that
\begin{equation}
\mathcal{D}\log(D_t^{\text{Corp},\epsilon})+r_t^{\text{Corp},\epsilon}=\mathcal{D}\log(\underbrace{D_t^{\text{Gov}}}_{\equiv1})+\underbrace{r_t^{\text{Gov}}}_{=0}=0,
\end{equation}
which, inserted into equation (\ref{eqDD}), leads to
\begin{equation}
\alpha_t^{\text{Corp},\epsilon}+r_t^{\text{Corp},\epsilon}+\frac{1}{2}\sigma_t^{\text{Corp},\epsilon}{\sigma_t^{\text{Corp},\epsilon}}^{\dagger}+\sigma_t^{\text{Corp},\epsilon}\frac{W_t}{2t}=0.
\end{equation}
By applying on both sides $\lim_{s\rightarrow t^+}\mathbb{E}_s[\cdot]$ we conclude that
\begin{equation}\label{mean0}
\alpha_t^{\text{Corp},\epsilon}+r_t^{\text{Corp},\epsilon}+\frac{1}{2}\sigma_t^{\text{Corp},\epsilon}{\sigma_t^{\text{Corp},\epsilon}}^{\dagger}=0.
\end{equation}
Hence, the integrand in formula (\ref{Novikov}) reads, being the government asset is the num\'{e}raire,
\begin{equation}
\left|{\sigma_t^{\text{Corp},\epsilon}}^{\dagger}\left(\sigma_t^{\text{Corp},\epsilon}{\sigma_t^{\text{Corp},\epsilon}}^{\dagger}\right)^{-1}\left(\alpha_t^{\text{Corp},\epsilon}+r_t^{\text{Corp},\epsilon}\right)\right|^2=\frac{1}{4}\sigma_t^{\text{Corp},\epsilon}{\sigma_t^{\text{Corp},\epsilon}}^{\dagger}.
\end{equation}
Equation (\ref{dlog}) becomes, after having inserted (\ref{mean0})
\begin{equation}
\left(\mathcal{D}\log(D_t^{\text{Corp},\epsilon})+r_t^{\text{Corp},\epsilon}\right)^2=
\sigma_u^{\text{Corp},\epsilon}{\sigma_u^{\text{Corp},\epsilon}}^{\dagger}\frac{W_t^{\dagger}W_t}{4t^2}.
\end{equation}
On the other hand, the second equation in (\ref{eq1}) leads to
\begin{equation}
\mathcal{D}\log(D_t^{\text{Corp}})=-\frac{\LGD_t\lambda_t}{1-\LGD_tX_t},
\end{equation}
and therefore
\begin{equation}
\lim_{\epsilon\rightarrow0}\left|{\sigma_t^{\text{Corp},\epsilon}}^{\dagger}\left(\sigma_t^{\text{Corp},\epsilon}{\sigma_t^{\text{Corp},\epsilon}}^{\dagger}\right)^{-1}\left(\alpha_t^{\text{Corp},\epsilon}+r_t^{\text{Corp},\epsilon}\right)\right|^2
=\left(\frac{\LGD_t\lambda_t}{1-\LGD_tX_t}-r_t^{\text{Corp}}\right)^2\frac{t}{Q_t(K)},
\end{equation}
which inserted into Novikov's condition (\ref{Novikov}) proves inequality (\ref{NovikovLGD}).\\
\end{proof}
\section{Credit Arbitrage Dynamics and Bubbles}
We apply now the results for arbitrage market bubbles as recalled in Subsection \ref{bubbles} to our simple credit market model consisting in two base asset, a government and a corporate bond. We allow arbitrage and just assume that the total quantity of potential arbitrage is minimized by the market forces, as explained in detail in \cite{Fa20,FaTa20Tris}.
\begin{theorem} [\textbf{Credit Arbitrage Dynamics and Arbitrage Bubbles for Credit Markets}]\label{CRThm} The following statements hold true for the credit market model of Definition \ref{CM}
\begin{equation}
D_t:=[D_t^{\text{Gov}}, D_t^{\text{Cred}}]^{\dagger}\qquad r_t:=[r_t^{\text{Gov}},r_t^{\text{Cred}}]^{\dagger}\qquad x_t:=[x_t^{\text{Gov}}, x_t^{\text{Cred}}]^{\dagger}
\end{equation}
\noindent with $T\le+\infty$ allowing for arbitrage:
\begin{itemize}
 \item[$(a)$] Market portfolio, asset values and term structures solving the minimal arbitrage problem are identically distributed along time, their returns are centered and serially uncorrelated:
 \begin{equation}\label{id2}
 \boxed{
 \begin{split}
 &([x_t;D_t;r_t]))_{t\in[0,T]}\text{ is an i.d. process with respect to  $P$,}\\
 &([\mathcal{D}x_t;\mathcal{D}D_t;\mathcal{D}r_t])_{t\in[0,T]}\text{ is centered and has vanishing autocovariance function,}
 \end{split}
 }
 \end{equation}
 \noindent In particular, conditional and total expectations of asset values, nominals and term structures are constant over time:
\begin{equation}
\boxed{
 \begin{array}{lll}
 \mathbb{E}_0[x_t]\equiv\text{const}\qquad &\mathbb{E}_0[D_t]\equiv\text{const}\qquad &\mathbb{E}_0[r_t]\equiv\text{const}\qquad\\
 \mathbb{E}_0[\mathcal{D}x_t]\equiv0\qquad &\mathbb{E}_0[\mathcal{D}D_t]\equiv0\qquad &\mathbb{E}_0[\mathcal{D}r_t]\equiv0.\qquad
 \end{array}
 }
 \end{equation}

 \noindent The variances of portfolio nominals are concurrent with those of the deflators:
\begin{equation}\label{vols}
\boxed{
\Var_0\left(D_t^j\right)\Var_0\left(\frac{x_t^j}{x_t\cdot D_t}\right)\ge\frac{1}{4},
}
\end{equation}
\noindent for all indices $j=1,2$.
  For the credit component we have:
 \begin{equation}
 \boxed{
  \mathbb{E}_0[x_t^{\text{Cred}}]\equiv\text{const}\qquad
  \mathbb{E}_0[(D^{\text{Corp}}_{t}]\equiv\text{const}\qquad
  \mathbb{E}_0[\lambda_t\LGD_t]\equiv\text{const}.
 }
 \end{equation}
 \noindent The variances of portfolio nominals are concurrent with those of the instantaneous short rates:
 \begin{equation}\label{vol_lower}
 \boxed{
\Var_0\left(D_t^{\text{Cred}}\right)\Var_0\left(\frac{x_t^{\text{Cred}}}{x_t^{\text{Gov}} D_t^{\text{Gov}}+x_t^{\text{Cred}} D_t^{\text{Cred}}}\right)\ge\frac{1}{4},
}
 \end{equation}

 \item[$(b)$] Expectation and variance of the discounted value for the credit bubble read
 \begin{equation}\label{exp_var_base_bubble}
 \boxed{
  \begin{split}
  \mathbb{E}_0[\widehat{B}_t^{\text{Cred}}]&=\mathbb{E}_0\left[D_t^{\text{Cred}}-\widehat{C}_t^{\text{Cred}}\right]-\mathbb{E}_0^*\left[D_T^{\text{Cred}} - \widehat{C}_T^{\text{Cred}}\right]\\
  \Var_0(\widehat{B}_t^{\text{Cred}})&=\Var_0\left(D_t^{\text{Cred}}-\widehat{C}_t^{\text{Cred}}\right)+\Var_0^*\left(D_T^{\text{Cred}}-\widehat{C}_T^{\text{Cred}}\right)+\\
  &\qquad-2\Cov_0^*\left(D_t^{\text{Cred}}-\widehat{C}_t^{\text{Cred}},D_T^{\text{Cred}}-\widehat{C}_T^{\text{Cred}}\right).
  \end{split}
  }
 \end{equation}

 \item[$(c)$] Expectation and variance of the discounted value for the credit derivative $G(S_T^{\text{Cred}})$ on the credit asset read
 \begin{equation}\label{exp_var_contingent_bubble}
 \boxed{
  \begin{split}
  \mathbb{E}_0[\widehat{B}_t^{\text{Cred}}(G)]&=\mathbb{E}_0[\widehat{V}_t^{\text{Cred}}(G)]-\mathbb{E}_0^*\left[\widehat{G}\left(S_T^{\text{Cred}}\exp\left(\frac{C_T^{\text{Cred}}}{S_T^{\text{Cred}}}(T-t)\right)\right)\right]\\
  \Var_0(\widehat{B}_t^{\text{Cred}}(G))&=\Var_0(\widehat{V}_t^{\text{Cred}}(G))+\Var_0^*\left(\widehat{G}\left(S_T\exp\left(\frac{C_T^{\text{Cred}}}{S_T^{\text{Cred}}}(T-t)\right)\right)\right).
  \end{split}
  }
 \end{equation}
\end{itemize}
\end{theorem}

\section{Conclusion}
By introducing an appropriate stochastic differential geometric
formalism, the classical theory of stochastic finance can be
embedded into a conceptual framework called Geometric
Arbitrage Theory, where the market is modelled with a principal
fibre bundle and arbitrage corresponds to its curvature. The tools developed can be applied to
default risk and recovery modeling leading to arbitrage and no arbitrage characterizations for credit markets,
as well as the explicit computation for arbitrage credit bubbles and credit market dynamics.







\end{document}